\documentclass[12pt,draftclsnofoot,onecolumn]{IEEEtran}
\usepackage{amsfonts}
\usepackage{amssymb}
\usepackage{amsmath}
\usepackage{mathrsfs}
\usepackage{graphicx}
\usepackage{verbatim}
\usepackage{subfigure}
\usepackage{balance}
\usepackage{booktabs}
\usepackage{fancybox}
\usepackage{bm}
\usepackage{extarrows}
\usepackage{algorithm}
\usepackage{algorithmic}
\usepackage{multirow}
\usepackage{array}
\usepackage{epstopdf}
\usepackage{caption}
\usepackage{longtable}
\usepackage{rotating}
\usepackage{multirow}
\usepackage{diagbox}
\usepackage{graphicx,epstopdf}

\newtheorem{corollary}{Corollary}
\newtheorem{proof}{Proof}

\def \F {{}_2F_1}

\begin{document}

\baselineskip 4.2ex

\title{Cooperative Secure Transmission by Exploiting Social Ties in Random Networks
\thanks{This work was partially presented at the IEEE Global Telecommunications Conference Workshop on Trusted Communications with Physical Layer Security (Globecom'17 - Workshop - TCPLS), Singapore, Dec. 2017 \cite{Xu2017}.
	\emph{(Corresponding author: Hui-Ming Wang).}
}
\thanks{$^\dagger$ The authors are with the
School of Electronics and Information Engineering, and also with the Ministry
of Education Key Lab for Intelligent Networks and Network Security, Xi'an Jiaotong
University, Xi'an, 710049, Shaanxi, P. R. China. Email: {\tt
ymxu29@outlook.com, xjbswhm@gmail.com, xjtu-huangkw@outlook.com}.
}
\thanks{$^\ddagger$ The author is with the Department of Electrical and Computer Engineering, University of Houston, Houston, TX 77004, USA. Email: {\tt zhan2@uh.edu}.}
\thanks{$^\S$ The author is with the School of Engineering, Nazarbayev University, Astana 010000, Kazakhstan. Email: {\tt theodoros.tsiftsis@nu.edu.kz}.}
\author{  Hui-Ming Wang$^\dagger$, \emph{Senior Member, IEEE}, \hspace{0.05in} Yiming Xu$^\dagger$, \hspace{0.05in} Ke-Wen Huang$^\dagger$, \hspace{0.05in} Zhu Han$^\ddagger$, \emph{Fellow, IEEE}, \hspace{0.05in} Theodoros A. Tsiftsis$^\S$, \emph{Senior Member, IEEE}
}
\vspace{-2mm}
}

\maketitle

\begin{abstract}
Social awareness and social ties are becoming increasingly popular with emerging mobile and handheld devices. Social trust degree describing the strength of the social ties has drawn lots of research interests in many fields in wireless communications, such as resource sharing, cooperative communication and so on. In this paper, we propose a hybrid cooperative beamforming and jamming scheme to secure communication based on the social trust degree under a stochastic geometry framework. The friendly nodes are categorized into relays and jammers according to their locations and social trust degrees with the source node. We aim to analyze the involved connection outage probability (COP) and secrecy outage probability (SOP) of the performance in the networks. To achieve this target, we propose a double Gamma ratio (DGR) approach through Gamma approximation. Based on this, the COP and SOP are tractably obtained in closed-form. We further consider the SOP in the presence of Poisson Point Process (PPP) distributed eavesdroppers and derive an upper bound. The simulation results verify our theoretical findings, and validate that the social trust degree has dramatic influences on the security performance in the networks.
\end{abstract}


\IEEEpeerreviewmaketitle

\newpage
\section{Introduction}
Nowadays, social ties have brought extensive influences among humankind. More and more people are actively involved in online social interactions \cite{SurveySocial1}, \cite{SocialTie}, and hence social ties among people are extensively broadened and significantly enhanced \cite{SocialTie4}. The so-called social ties are usually defined as the social relationships between individuals \cite{DefineSocialTie}, such as kinship, colleague relationships, friendship, acquaintance, and so on \cite{SocialAware2}. The social trust degree of the social tie is the most basic and fundamental notion which characterizes the strength of two individuals relating to each other \cite{SocialAware}. According to \cite{MeasureSocial}, ties have specific trust degree values describing the strength (i.e., from enmity to kinship) between the users. The notions of social ties and social trust degrees have drawn wide attentions of researchers in various fields, including mobile social networks, wireless network communications, and so on. For instance, social ties have also been studied for cooperative communications \cite{D2DSocial3}--\cite{RelaySocial3}. The authors in \cite{D2DSocial3} investigated the joint social-position relationship based cooperation (JSPC) scheme and developed a partner selection algorithm. An optimal social-aware relay selection strategy was proposed in \cite{RelaySocial2} to maximize the capacity of the network. An optimal transmission beamformer design was considered in \cite{RelaySocial3} based on the trust degrees to achieve a target rate in a multi-input single-output cooperative communication network.
These existing work concentrates on efficiency and capacity analysis in these networks. However, the security and privacy of information is a significant issue in social awareness networks. Due to the openness of wireless communications, the leakage of information is a serious problem.

On the other hand, physical Layer Security (PLS) approaches have drawn considerable attention during the past decade to protect the confidentiality of wireless communications. Wyner's seminal research in \cite{Wyner} introduced the concept of the wiretap channel and secrecy capacity, and established a basic theory for PLS. According to the Wyner's theory, a positive secrecy capacity exists if the channel quality of the legitimate receiver is better than that of the eavesdropper. To achieve this target, various PLS communication technologies have been proposed, among which the multiple-user cooperation technology has been studied intensively.
As indicated by the survey paper \cite{BossTutarial}, the cooperative beamforming (CB) and the cooperative jamming (CJ) are two effective methods to improve PLS.
Out of a bunch of friendly nodes in the network, some nodes are selected as relay nodes and other as jammers. Then, relay nodes exploit CB \cite{CooBeamforming1}--\cite{CooBeamforming5} to assist enhancing the channel quantity of the legitimate users, while jammers utilize the CJ \cite{Jam3}, \cite{Jam5} to degrade that of the eavesdroppers.
The friendly nodes in the cooperative networks have been categorized into relays and jammers in these existing works. However, most of them assumed that the relays or jammers are selected and assigned the roles they act (relay or jammer), but did not discuss how the assignment of the roles was made.

Motivated by the above researches, we observe that the social trust degree may play an important role in cooperative secure communications. The social ties of users reflect their willingness to share resources in order to help secure communications. The users with high social trust degrees usually have strong ties and are willing to share resources to help each other for secret communications. It is reasonable that the strong-tie nodes are more likely to offer communication links than the weak-tie ones for cooperation. Particularly, in secure cooperative communications, two nodes with high social trust degrees may have a high probability to establish connections, and to decode or retransmit the confidential messages without leaking them to potential eavesdroppers. Based on the above background, the social trust degrees among users can be utilized to assist cooperative communications for PLS. Motivated by these observations, in this paper, we propose a cooperative relay and jamming scheme to secure communication based on the social trust degree. The security performance is investigated by connection outage probability (COP) and secrecy outage probability (SOP) under a stochastic geometry framework.

\subsection{Related Works}
Researches on social ties have been carried out in many aspects to enhance the efficiency of wireless communications \cite{SocialTie}, \cite{D2DSocial}--\cite{AdHocSocial2}. Utilizing of social ties has been discussed in \cite{SocialTie}, \cite{D2DSocial}--\cite{D2DSocial5} to enhance the performance of device-to-device (D2D) communications. Abundant frameworks and approaches based on social ties such as the coalitional game-theoretic framework \cite{SocialTie}, the Indian Buffet Process approach \cite{D2DSocial}, several resource allocation policies \cite{D2DSocial4,D2DSocial5}, and so on, were proposed to optimize the traffic offloading and improve the system capacity. The efficiency and capacity of wireless ad hoc networks have been improved by exploiting the social ties among users \cite{AdHocSocial}, \cite{AdHocSocial2}.

In cooperative PLS, to achieve a larger secrecy rate, various CB and CJ schemes have been proposed in \cite{JointInAF}--\cite{Hybird2}. The joint cooperative jamming and beamforming schemes were proposed in \cite{JointInAF}--\cite{JointRobustAF}, and were further developed to a destination assisted scheme in \cite{JointPositionAF}. Furthermore, the joint CB and CJ schemes were investigated in hybird networks \cite{HybirdRelay,Hybird2}.

So far, there are few works to study social ties among users in cooperative communications for PLS enhancement \cite{SecrecySocial1}--\cite{SecrecySocial4}.
Zheng \textit{et al.} \cite{SecrecySocial1} studied the secrecy rate and the secrecy throughput under a multi-hop relay scheme using the average source-destination distance based on social ties.
To further enhance security,
both \cite{SecrecySocial2} and \cite{SecrecySocial3} proposed cooperative jamming schemes based on social ties.
Tang \textit{et al.} \cite{SecrecySocial2} discussed the SOP of a source-destination pair based cooperative jamming game. A jammer selection scheme based on mobility-impacted social interactions was proposed in \cite{SecrecySocial3} to maximize the worst-case ergodic secrecy rate. When the relays may be potential eavesdroppers according to their social trust degrees,
a cooperative communication strategy was presented in \cite{SecrecySocial4} to maximize the secrecy rate.
We note that in these works, the cooperative nodes are either relay nodes or jammers, and a joint scheme is missing, which means that the secrecy performance can be improved further by improving the cooperative strategy. Moreover, the social trust degree is merely applied when the relays or jammers have already been chosen, i.e., the social trust degree has not been exploited to determine which role of each cooperative node should be categorized.

\subsection{Our Work and Contributions}
In this paper, we propose a hybrid cooperative relay and jamming scheme exploiting social ties to secure wireless communications in a random cooperative network. The friendly nodes are categorized into relays and jammers according to their locations and social trust degrees with the source node. We analyze the COP and SOP to evaluate the security performance under a stochastic geometry framework.
To the best of our knowledge, this is the first paper that applies the social trust degree into cooperative node categorization and hybrid cooperative secrecy communications.
Our contributions are summarized as follows:

\textbf{i)}
We propose a social trust degree based hybrid cooperative beamforming and jamming scheme to secure a wireless transmission under a stochastic geometry framework, wherein the cooperative nodes are spatially distributed in a two dimensional plane following a Poisson Point Process (PPP). According to the social trust degrees of the source, the cooperative nodes are categorized into relays and jammers. In general, the hybrid cooperative scheme is distributed with a low cooperative overhead.

\textbf{ii)}
A comprehensive COP and SOP analysis is performed to evaluate the performance of the proposed scheme. To facilitate convenient results with a sufficient accuracy, we propose a Gamma approximation method and a double Gamma ratio (DGR) approach to provide closed-form expressions of the COP and the SOP. In terms of the derivation of parameters utilizing Gamma approximation method, a
flabellate annulus approximation method is also proposed to simplify the complicated integral calculations over an irregular pattern.

\textbf{iii)}
As an extension, we further investigate the SOP in the presence of independent and homogeneous PPP distributed eavesdroppers, and obtain its upper bound. In the derivation of the SOP with multiple eavesdroppers, three independent PPPs are contained which makes the calculations untractable. In order to obtain an upper bound, the discrete expectation utilizing the law of total probability is employed to approximate the continuous expectation.

\subsection{Organization and Notations}
The reminder of this paper is organized as follows. In Section II, we provide our system model with relay and jammer selection schemes based on the social trust degrees of the source node. In Section III, we propose a DGR approach based on the Gamma approximation method, which provide general formulations for simplifying calculations in our analysis. In Section IV, we investigate the COP in our scheme. In Section V, the SOP with single eavesdropper and PPP distributed eavesdroppers are analyzed, respectively. In Section VI, we provide the numerical results to verify our theoretical analysis and illustrate the performance of the proposed scheme. Finally, Section VII concludes the paper.

We use the following notations in this paper: $(\cdot)^*$, $\|\cdot\|$ and $|\cdot|$ denote conjugate, Euclidean norm, and absolute value, respectively. $\mathbb{P}(\cdot)$ denotes probability. $\mathbb{E}_A[\cdot]$ and $\mathbb{D}_A[\cdot]$ denote mathematical expectation and variance with respect to $A$, respectively. $\mathcal{CN}(\mu, \sigma^2)$ denotes circularly symmetric complex Gaussian distribution with mean $\mu$ and variance $\sigma^2$. exp(1) denotes exponential distribution with mean 1. $\mathcal{A}(x,r)\subset \mathbb{R}^2$ denotes a bi-dimensional disk centered $x$ with radius $r$, and $\mathcal{D}(L_1,L_2)\subset \mathbb{R}^2$ denotes an annulus with internal radius $L_1$ and external radius $L_2$.

\section{System Model and Problem Description}
As illustrated in Fig.~\ref{SCP_M_PHI}, we consider a wireless network over a finite circle area $\mathcal{A}(o,L_2)\subset \mathbb{R}^2$. This network consists of one source s, one destination $y$, one eavesdropper $z$, and lots of legitimate nodes. Each node in the network works in a half-duplex mode and is equipped with a single antenna. Without loss of generality, we assume that source s locates at the origin $(0,0)$. The source tends to transmit confidential signals to the destination node $y$ without being wiretapped by the eavesdropper. To achieve this target, the source hopes that the legitimate nodes can help complete the secure transmission by node cooperation.

\begin{figure}[!t]
	\centering
	\includegraphics[width=5in]{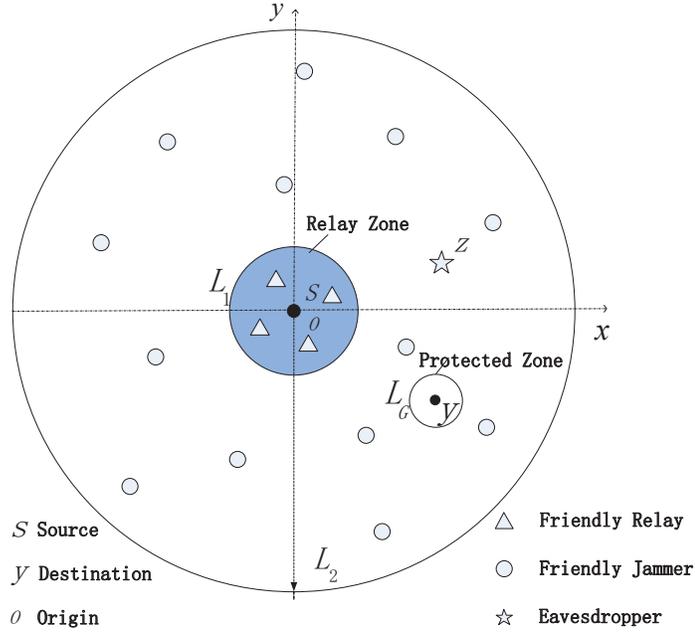}
	\caption{System model.}
	\label{SCP_M_PHI}
\end{figure}

The location of the legitimate node is modeled as a PPP $\Phi$ with density $\lambda$. Each legitimate node has a social tie with the social trust degree of the source, such as family members, colleagues, friends, and so on. The degree between two individuals is usually modeled by using a value in the range $[0,1]$, which is similar to \cite{MeasureSocial}, \cite{D2DSocial3}, \cite{RelaySocial3}, \cite{SecrecySocial4}.\footnote{~A classical way of modelling the social ties between two individuals is weighted graphs which can be referred to \cite{MeasureSocial}, where each node represents one person. The strength of the interactions between individuals is represented by the weights associated with each edge. In further researches, the strength is modelled as a value in the range [0, 1]. } We assume that the social trust degree of each legitimate node is independent identically distributed (i.i.d.), modeled as a uniform random variable (RV) $C$ distributed in $[0,1]$. For a legitimate node, the closer $C$ near to 1, the more the source trust in this node.

\subsection{Social Aware Nodes Selection}
In this paper, we propose a secure cooperative transmission scheme, where these legitimate nodes help increase the secrecy rate via cooperative beamforming and jamming.
We categorize the legitimate nodes as relay nodes, jammers and dummy nodes according to their social trust degrees of the source as well as their locations.

\begin{itemize}
  \item Relay nodes: Intuitively, the relay nodes should have the closest social ties or the most trustworthiness of the source since they may have the permission to  relay the confidential signals or even decode them. Therefore, the most trusted nodes are selected as relay nodes if their social trust degrees are in $[C_1,1]$, where $0<C_1<1$ is a sufficiently large threshold. On the other hand, relay nodes should be close to the source geographically so that the confidential signal broadcasted in the first cooperative phase will not be leaked to the potential eavesdropper. Therefore, we also require the relay nodes be located in $\mathcal{A}(o,L_1)$, i.e., only those legitimate nodes with a distance less than $L_1$ to the source are possible to be relay nodes.
  \item Jammers: Jammers are friendly nodes to the source but are not the most trusted ones. They will not help relay the confidential signal, but transmit artificial interferences to disturb the potential eavesdropper when the confidential signal transmission is ongoing. Therefore, we set the legitimate nodes whose social trust degrees are in $[C_2,C_1]$ as jammers, where $0<C_2<C_1$.
  \item Dummy nodes: Those nodes whose social trust degrees are in $[0,C_2]$ are dummy nodes. Their social connections are not tight with the source. They will not take part in the confidential transmission of the source and do nothing.
\end{itemize}

As a result, those legitimate nodes who will help the secrecy transmissions are divided into relays and jammers according to their social trust degrees and locations, where the locations of the relays in $\mathcal{A}(o,L_1)$ and the jammers in $\mathcal{D}(L_1,L_2)$ are  two independent and homogeneous PPPs, $\Phi_R$ and $\Phi_J$ with intensities $\lambda_R = (1 - C_1)\lambda$ and $\lambda_J = (C_1 - C_2)\lambda$, respectively. Throughout this paper, we will use $x_R\in\Phi_R$ to denote the relay and $x_J\in\Phi_J$ to denote the jammer. Such a system model can be easily found in the scenes like offices, dormitories, labs and so on, where people inside the room may be trusted so that they are reliable and willing to assistant secure transmission, while people outside may help transmit jamming signals to disturb eavesdroppers.

\subsection{Secure Cooperative Transmission Scheme}
We assume that all the relays work in the decode-and-forward (DF) mode. During the first phase of secure cooperative communication, the source $s$ broadcasts confidential signal to the relays in $\mathcal{A}(o,L_1)$. In this phase, we assume that the relays can always decode the received confidential information correctly and the confidential information can be transmitted securely, which is due to the following reasons. Since the relays are required to be located by the source with a short distance less than $L_1$, the source transmits with a sufficiently low power, such that the signal can be correctly decoded by the relays. The signal can not be decoded correctly by the eavesdropper outside $\mathcal{A}(o,L_1)$ due to the large-scale pass loss. Such assumption comes from some scenarios where the source user stays along with several legitimate friendly users in a region, such as colleagues, workmates, roommates in the workplace or labs. When the source user aim to transmit secure information to a destination far away, he or she may select the legitimate friendly users as relays in the region according to their social trust degrees. The security assumption of the source-relay link has also been adopted in \cite{CooBeamforming1}, \cite{CooBeamforming5}, \cite{JointRobustAF}.

During the second phase of cooperation, the relays in $\mathcal{A}(o,L_1)$ forward the correctly decoded confidential information bits to the destination cooperatively. Since the destination is far away from the relays, the risk of being wiretapped in the second phase increases greatly. Therefore, the jammers in $\mathcal{D}(L_1,L_2)$ transmit jamming signals concurrently to confuse the eavesdropper.
In order to protect the signal received at destination node $y$ from being disturbed, the jammers near the destination should keep silence. Therefore, we set a protected zone $\mathcal{A}(y,L_G)$, i.e., a circle of radius $L_G$ centered at $y$, wherein the jammers will keep silence during the second phase  \cite{ProtectedZone}.\footnote{~Destination node $y$ first broadcasts a pilot signal with a pre-designed power. If a jammer received the pilot signal, it is in the protected zone and it will not transmit jamming signals. As a result, the received signal at the destination will be protected. }

\subsection{Channel and Signal Model}
The signal suffers from both small-scale fading and large-scale path loss. We assume that the small-scale fading is quasi-static following the Rayleigh distribution, and the channel coefficient between two nodes located in $x_1$ and $x_2$ is denoted as $H_{x_1,x_2}\thicksim \mathcal{CN}(0,1)$.
The large-scale fading is standard path loss model $d^{-\alpha}$, where $d$ denotes the distance and $\alpha>2$ is the fading power exponent \cite{Channel1}.

Since the relays are required to be located with a short distance by the source, 
the relay nodes that correctly decode the signal $s_0$ will forward it to the destination cooperatively in the second phase.
In this paper, we will consider a distributed cooperative beamforming scheme, where each relay transmits $s_0$ by pre-compensating the phase of the channel $H_{x_R,y}$ and utilizing all its available power. The transmitted signal of each relay node is $s_{x_R} = \frac{\sqrt{P_s}H_{x_R,y}^*}{\|H_{x_R,y}\|}s_0$, where the symbol power is normalized as $E\{|s_0|^2\}=1$ and the transmit power of the relay is $P_R$. We note that such a cooperative beamforming scheme is totally distributed in the sense that each relay do the cooperation with its own channel state information (CSI) instead of the global CSI, so that the overhead is greatly reduced.

Then, the received signal power at destination node $y$ is given by
\begin{align}
T(y) & = \left|\sum_{x_R\in\Phi_R}\frac{\sqrt{P_R}H_{x_R,y}H_{x_R,y}^*}{\|H_{x_R,y}\|}\|x_R - y\|^{-\alpha/2}\right|^2 \nonumber \\
& = \left|\sum_{x_R\in\Phi_R}\sqrt{P_R}\|H_{x_R,y}\|\|x_R - y\|^{-\alpha/2}\right|^2 .\label{Ty}
\end{align}
Similarly, the signal power received by the eavesdropper is given by
\begin{align}
T(z) & = \left|\sum_{x_R\in\Phi_R}\frac{\sqrt{P_R}H_{x_R,z}H_{x_R,y}^*}{\|H_{x_R,y}\|}\|x_R - z\|^{-\alpha/2}\right|^2 .\label{Tz}
\end{align}
When the cooperative beamforming is ongoing, each jammer also transmits an independent Gaussian interference signal to confuse the eavesdropper. The transmit power of the jammer is $P_j$. Since the jamming signals from different jammers are independent, the aggregate interference power at destination $y$ and eavesdropper $z$ are given by
\begin{align}
I(y) = \sum_{x_J\in\mathcal{\overline{D}}}P_jh_{x_J,y}\| x_J - y\|^{-\alpha}, \label{Iy}
\end{align}
and
\begin{align}
I(z) = \sum_{x_J\in\mathcal{\overline{D}}}P_jh_{x_J,z}\| x_J - z\|^{-\alpha}, \label{Iz}
\end{align}
respectively, where $\mathcal{\overline{D}}$ denotes the area $\Phi_J\backslash\mathcal{A}(y,L_G)$, and $h_{x_1,x_2}\thicksim \exp(1)$ is the power fading between locations $x_1$ and $x_2$. Suppose the network is interference-limited so that the ambient noise is negligible. The signal-to-interference ratio (SIR) for destination $y$ is given by
\begin{align}
SIR_y = \frac{\left|\sum_{x_R\in\Phi_R}\sqrt{P_R}\|H_{x_R,y}\|\|x_R - y\|^{-\alpha/2}\right|^2}{\sum_{x_J\in\Phi_J\backslash\mathcal{A}(y,L_G)}P_jh_{x_J,y}\| x_J - y\|^{-\alpha}}, \label{SIRy}
\end{align}
and the SIR for eavesdropper $z$ is given by
\begin{align}
SIR_z = \frac{\left|\sum_{x_R\in\Phi_R}\frac{\sqrt{P_R}H_{x_R,z}H_{x_R,y}^*}{\|H_{x_R,y}\|}\|x_R - z\|^{-\alpha/2}\right|^2}{\sum_{x_J\in\Phi_J\backslash\mathcal{A}(y,L_G)}P_jh_{x_J,z}\| x_J - z\|^{-\alpha}}. \label{SIRz}
\end{align}
In order to simplify our expressions, we define $d_{x,y} \triangleq \|x - y\|$ to denote the distance between the nodes located at $x$ and $y$.

\section{Gamma Approximation and DGR Approach} \label{SectionDGR}
In this paper, we aim to analyze the COP and the SOP in our system according to \eqref{SIRy} and \eqref{SIRz} in Section II. The detailed derivations of the COP and the SOP are provided in Section IV and Section V, respectively. We will show that the expressions of COP and SOP \cite{SOPbyOther}, \cite{SOPbyTXZheng} have a unified formulation as
\begin{align}
\mathcal{P} & = \mathbb{P}\left(SIR<\beta\right) = \mathbb{P}\left(\frac{T}{I}<\beta\right) ,\label{OP}
\end{align}
where $T$ is the signal power, $I$ is the interference power, and $\beta$ is some threshold according to the target performance of the system. Such a formulation is widely used in the calculation of the outage in PLS. In our scheme, $T$ and $I$ have complicated distributions without closed-form expressions of the probability distribution functions (PDF), which makes our analysis quite untractable. To facilitate the analysis, in this section, we propose the following two calculation methods, i.e., the Gamma approximation method \cite{GamaAppr} and a DGR approach. We first introduce the Gamma approximation method.

\subsection{Gamma Approximation}
Gamma approximation is a model approach to approximate the distribution of a RV based on the Gamma distribution, which aims to facilitate simplified and low complexity calculations. We can present simplified expressions for the outage probabilities by using the Gamma approximation. For a RV $\bar{A}$, we use a Gamma RV $A$ with the PDF
\begin{align}
G_A(x_{A};\nu_{A},\theta_{A}) & = \frac{x_{A}^{\nu_{A} - 1}e^{-\frac{x_{A}}{\theta_{A}}}}{\theta_{A}^{\nu_{A}}\Gamma(\nu_{A})} \label{Gama}
\end{align}
to approximate it, where $\Gamma(\nu_{A}) = \int_0^\infty m^{\nu_{A}-1}e^{-m}dm$, and  $\nu_{A}$ and $\theta_{A}$ are derived from the cumulants of  $\bar A$, especially the mean and variance. The $i$-th cumulants $N_{\bar A}^{(i)}$ of a RV $\bar A$ is defined as
\begin{align}
N_{\bar A}^{(i)} & = \frac{d^{i}\mathbb{E}_{\bar A}\left[e^{w\bar a}\right]}{dw^{i}}\Big|_{w = 0} .\label{Eti}
\end{align}
The mean of $\bar A$ denoted as $\mu_{\bar A} = N_{\bar A}^{(1)}$, and the variance of $\bar A$ is $\sigma_{\bar A}^2 = N_{\bar A}^{(2)} - \left(N_{\bar A}^{(1)}\right)^2$. Then the corresponding function of the Gamma distribution in (\ref{Gama}) has the parameters
\begin{align}
\nu_A = \frac{\mu_{\bar A}^2}{\sigma_{\bar A}^2}\quad\mathrm{and}\quad\theta_A = \frac{\sigma_{\bar A}^2}{\mu_{\bar A}} .\label{vt}
\end{align}
Consequently, we can obtain the Gamma approximations $G_T(x_T;\nu_T,\theta_T)$ and $G_I(x_I;\nu_I,\theta_I)$ of $T$ and $I$ in our system model according to \eqref{Gama}-\eqref{vt}, respectively. The details will be provided in the following sections.

Gamma approximation can be applied to approximately model the sum of several variables with special distributions, such as the Rayleigh distribution and so on. The PDF of the sum of these variables is usually complicated to obtain, while the distribution of each variable can be modeled as a special case of the Gamma distribution \cite{GamaAppr}. According to the additivity of the Gamma distribution, the sum of several Gamma variables is still Gamma distributed. Consequently, the sum of such variables can be approximately modeled by a Gamma distribution. In our proposed system model, the PDFs of \eqref{Ty}--\eqref{Iz} are untractable to obtain. Fortunately, each item in the sum expression is a special case of the Gamma distribution. We use Gamma approximation to approximately modeled them, i.e., both the numerator and the denominator in the SIR can be modeled as Gamma variables. Namely, our objective outage probability in \eqref{OP} has the form of a DGR. Next, we will introduce the proposed DGR approach.

\subsection{The DGR Approach}
The DGR approach is to provide a convenient calculation of the cumulative distribution function of the ratio of two Gamma random variables. Given $T\thicksim G_T(x_T;\nu_T,\theta_T)$ and $I\thicksim G_I(x_I;\nu_I,\theta_I)$, we have the following corollary.

\begin{corollary}
The cumulative distribution function (CDF) of a ratio Gamma variables is given by
\begin{align}
\mathbb{P}\left(\frac{T}{I}<\beta\right) & =  1 - \frac{q^{\nu_T}\Gamma(\nu_T + \nu_I)}{\nu_I(q + 1)^{\nu_T + \nu_I}\Gamma(\nu_T)\Gamma(\nu_I)}~\F\left(1,\nu_T + \nu_I;\nu_I + 1;\frac{1}{q + 1}\right),
\end{align}
where $q = \frac{\beta\theta_I}{\theta_T}$, $\beta$ is a threshold, and ${}_2F_1\left(a,b;c;d\right)$ denotes hypergeometric function \cite[Eq. 6.455.1]{TableIntegral}.
\end{corollary}

\begin{proof}
\begin{align}
\mathbb{P}\left(\frac{T}{I}<\beta\right) & = \mathbb{P}\left(T<\beta I\right) \nonumber \\
& \overset{(a)}{=} \mathbb{E}_I\left[\int_0^{\beta I}G_T(x_T;\nu_T,\theta_T)dx_T\right] \nonumber \\
& \overset{(b)}{=} \int_0^{\infty}\left(1 - \frac{\Gamma\left(\nu_T,\frac{\beta I}{\theta_T}\right)}{\Gamma(\nu_T)}\right)G_I(x_I;\nu_I,\theta_I)dx_I \nonumber \\
& = \int_0^{\infty}G_I(x_I;\nu_I,\theta_I)dx_I - \int_0^{\infty}\frac{\Gamma\left(\nu_T,\frac{\beta I}{\theta_T}\right)}{\Gamma(\nu_T)}G_I(x_I;\nu_I,\theta_I)dx_I \nonumber \\
& \overset{(c)}{=} 1 - \frac{q^{\nu_T}\Gamma(\nu_T + \nu_I)}{\nu_I(q + 1)^{\nu_T + \nu_I}\Gamma(\nu_T)\Gamma(\nu_I)}~\F\left(1,\nu_T + \nu_I;\nu_I + 1;\frac{1}{q + 1}\right) ,\nonumber
\end{align}
where $(a)$ follows since $T\thicksim G_T(x_T;\nu_T,\theta_T)$, $(b)$ follows from substituting the definition of the incomplete gamma function \cite[Eq. 6.45]{TableIntegral} and the Gamma function of $I$. After some integral calculation, $(c)$ follows from applying \cite[Eq. 6.455.1]{TableIntegral}.
\end{proof}

As a result, we obtain the closed-form probability expression through the DGR approach for the ratio of the signal power and the interference power with the approximated Gamma distribution. Corollary 1 simplifies the calculation of the probability derived from the ratio of two Gamma variables. The applications and veracity of the Gamma approximation and DGR approach will be discussed in the following sections.

\section{Connection Outage Probability}
In this section, we will analyze the COP in our scheme based on the Gamma approximation and DGR approach, and provide its closed-form expression. 

Connection outage occurs when destination $y$ is unable to decode the signals transmitted by the relays, i.e., $SIR_y<\beta$ \cite{SOPbyOther}. The COP is given by
\begin{align}
\mathcal{P}_{to} & = \mathbb{P}\left(SIR_y<\beta\right) = \mathbb{P}\left(\frac{T(y)}{I(y)}<\beta\right) \nonumber \\
& = \mathbb{P}\left(\frac{\left|\sum_{x_R\in\Phi_R}\sqrt{P_R}\|H_{x_R,y}\|d_{x_R,y}^{-\alpha/2}\right|^2}{\sum_{x_J\in\mathcal{\overline{D}}}h_{x_J,y}d_{x_J,y}^{-\alpha}}<\beta\right) . \label{Pto}
\end{align}
According to \eqref{Ty} and \eqref{Iy} in Section II, for relay nodes at different locations, although $\|H_{x_R,y}\|$ follows the Rayleigh distribution, the values of the large scale fading $d_{x_R,y}^{-\alpha/2}$ are different. Consequently, $h_{x_R,y}d_{x_R,y}^{\alpha/2}$ with various locations $x_R$ are \emph{independent but not identically distributed}. Moreover, the means and the variances of them are random variables related to the random locations of the relays, which makes our analysis untractable. The case is similar for $I(y)$. As a result, it is untractable to obtain the accurate COP. Nevertheless, notice that both $T(y)$ and $I(y)$ are the sums of several Rayleigh or exponential random variables, which is approximated Gamma distributed. We can apply the Gamma approximation and DGR approach to facilitate our calculations.  We first model $T(y)$ and $I(y)$ using Gamma approximation.

The parameters of the PDFs of $T(y)$ (i.e. $\nu_{Ty}$, $\theta_{Ty}$) and $I(y)$ (i.e. $\nu_{Iy}$, $\theta_{Iy}$) derived from \eqref{Eti} and \eqref{vt} are given by the following proposition.

\textit{Proposition 1:} The PDFs of $T(y)$ and $I(y)$ have the parameters
\begin{align}
\nu_{Ty} = \frac{\lambda_RQ^2_y(1)}{5\lambda_RQ^2_y(1)+Q_y(2)},\quad
\theta_{Ty} = \frac{3P_R\left[5\lambda_RQ^2_y(1)+Q_y(2)\right]}{Q_y(1)},
\end{align}
\begin{align}
\nu_{Iy} = \frac{\lambda_J\Big(\int_\mathcal{\overline{D}}\frac{1}{d_{x_J,y}^{\alpha}}dx_J\Big)^2}{2\int_\mathcal{\overline{D}}\frac{1}{d_{x_J,y}^{2\alpha}}dx_J} 
,\quad\theta_{Iy} = \frac{2P_j\int_\mathcal{\overline{D}}\frac{1}{d_{x_J,y}^{2\alpha}}dx_J}{\int_\mathcal{\overline{D}}\frac{1}{d_{x_J,y}^{\alpha}}dx_J} ,\label{thetaI}
\end{align}
where $Q_y(n) = \int_{\mathcal{A}(o,L_1)}d_{x_R,y}^{-n\alpha}\mathrm{d}x_R$.
\begin{proof}
The proof of deriving $\nu_{Ty}$, $\theta_{Ty}$ and $\nu_{Iy}$, $\theta_{Iy}$ are given in Appendix A and B, respectively.
\end{proof}


\begin{figure}[!t]
\centering
\includegraphics[width=5in]{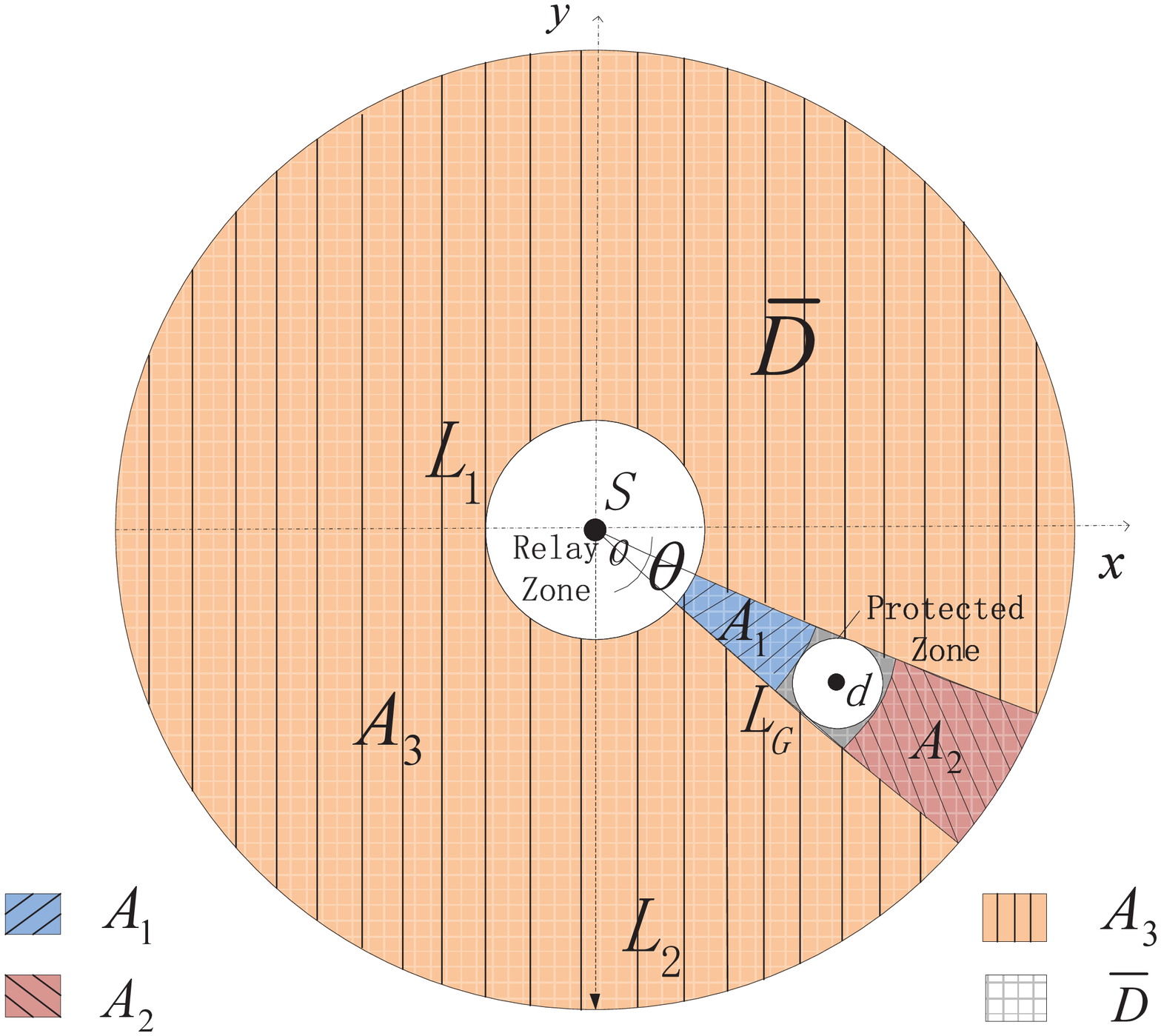}
\caption{ Integral area $\mathcal{\overline{D}}$ and the tractable approximation calculation}
\label{JamModel}
\end{figure}

Notice that due to the existence of the protected zone, the shadow integral area $\mathcal{\overline{D}}$ in \eqref{thetaI} illustrated in Fig.~\ref{JamModel} is untractable. We now propose a flabellate annulus approximation method to complete our calculations in \eqref{thetaI}.
As demonstrated in Fig.~\ref{JamModel},
the area $\mathcal{\overline{D}}$ in the shadow could be well approximated by the sum of the following three parts:
$A_1$ is the area with oblique line,  $A_2$ is in backslash and $A_3$ in vertical line is the annulus $\mathcal{D}(L_1,L_2)$ disposing the flabellate area with angles $\theta$.  Consequently, $\mathcal{\overline{D}}$ can be  approximated by the area $A_1 + A_2 + A_3$, i.e., we can use $\int_{A_1+A_2+A_3}f(x)dx$ to take the place of $\int_{\mathcal{\overline{D}}}f(x)dx$ in order to complete our integral calculations. Since $A_1$, $A_2$ and $A_3$ are flabellate annulus, it is quite convenient to calculate the integral in the area $A_1 + A_2 + A_3$. We name such a method as flabellate annulus approximation method.

According to (\ref{Gama}) and Proposition 1, the PDFs of $T(y)$ and $I(y)$ are given as
\begin{align}
f_{T(y)}(x_T;v_{Ty},\theta_{Ty}) = \frac{x_T^{v_{Ty} - 1}e^{-x_T/\theta_{Ty}}}{\theta_{Ty}^{v_{Ty}}\Gamma(v_{Ty})} \label{GamaT}
\end{align}
and
\begin{align}
f_{I(y)}(x_I;v_{Iy},\theta_{Iy}) = \frac{x_I^{v_{Iy} - 1}e^{-x_I/\theta_{Iy}}}{\theta_{Iy}^{v_{Iy}}\Gamma(v_{Iy})} ,\label{GamaI}
\end{align}
respectively.

\begin{figure}[!t]
\centering
\subfigure[Distribution of $T(y)$]{\label{TKf1}
\includegraphics[width=3.5in]{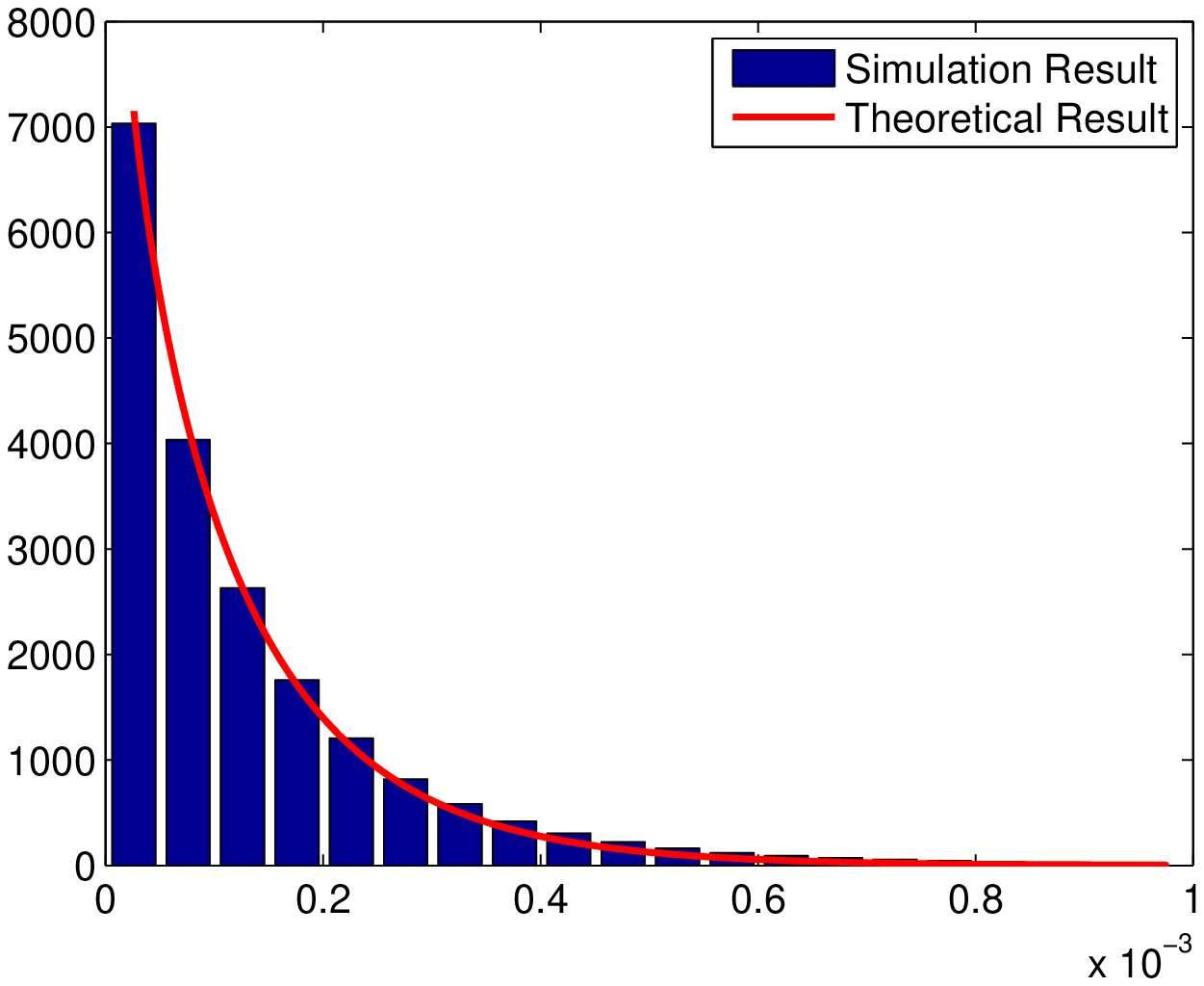}}
\subfigure[Distribution of $I(y)$]{\label{IKf1}
\includegraphics[width=3.5in]{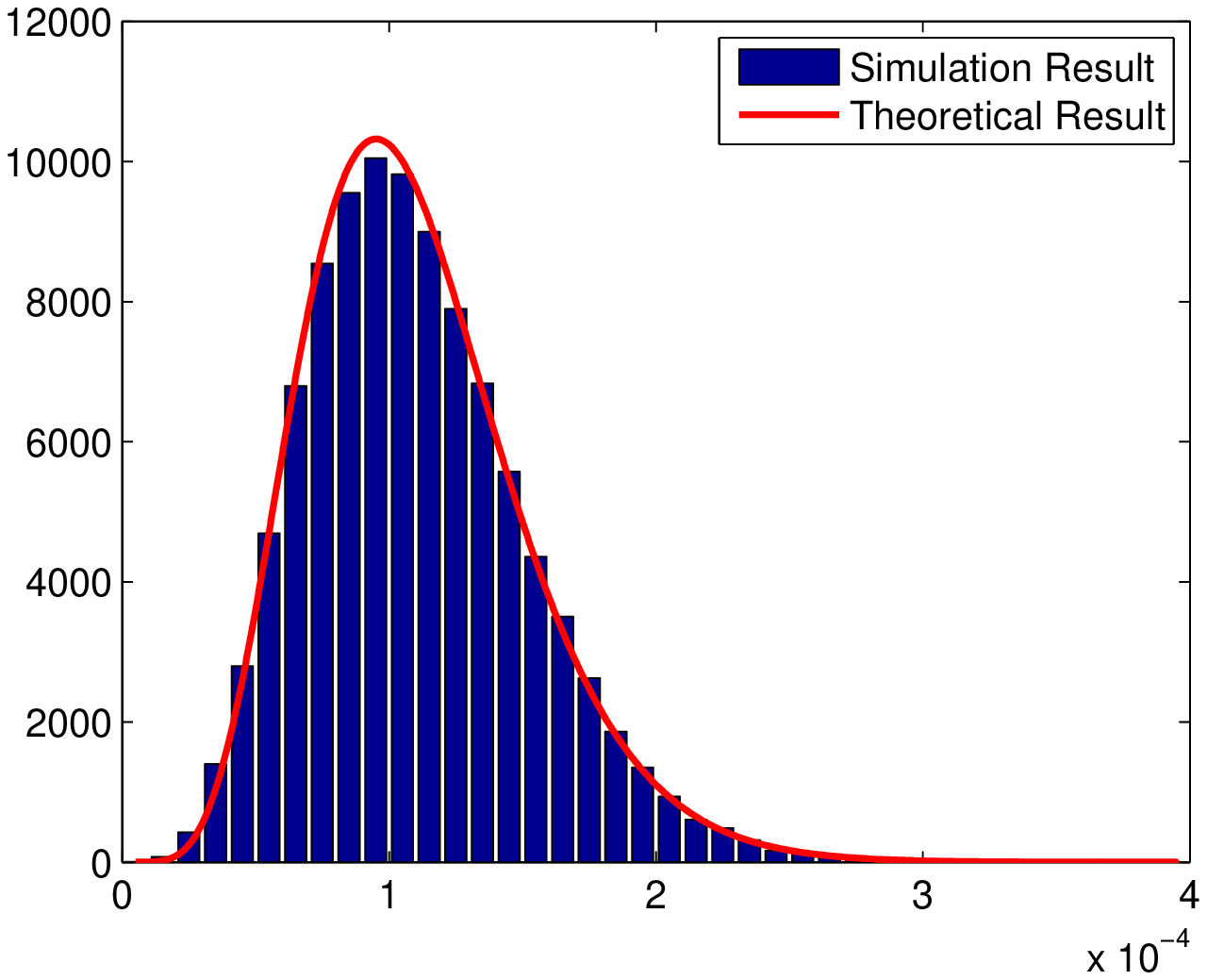}}
\caption{ Approximations based on the Gamma distribution for $T(y)$ and $I(y)$: (a) Simulation and theory results of the distribution of $T(y)$; (b) Simulation and theory results of the distribution of $I(y)$. The parameters are $\alpha=4$, $L_1=6$ m, $L_2=100$ m, $d=60$ m, $L_G=5$ m, $C_1=0.8$, $C_2=0.79$, $\lambda=0.2 /\mathrm{m}^2$, $P_R=10$ dBm, and $P_j=1$ dBm.}
\label{GammaModel}
\end{figure}

\begin{figure}[!t]
\centering
\includegraphics[width=4in]{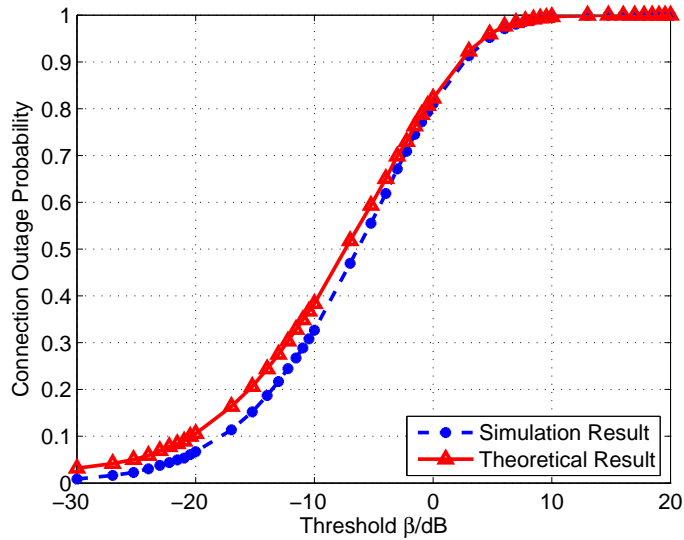}
\caption{Connection outage probability $\mathcal{P}_{to}$ vs. $\beta$ for our system, with $\alpha=4$, $L_1=6$ m, $L_2=100$ m, $d=60$ m, $L_G=5$ m, $C_1=0.8$, $C_2=0.79$, $\lambda=0.2 /\mathrm{m}^2$, $P_R=10$ dBm, and $P_j=1$ dBm.}
\label{COP1}
\end{figure}

Fig.~\ref{GammaModel} illustrates the accuracy of the Gamma approximation method. The PDFs of $T(y)$ and $I(y)$ are compared with the Gamma distribution in Fig.~\ref{TKf1} and Fig.~\ref{IKf1}, respectively. The theoretical results are derived from \eqref{GamaT} and \eqref{GamaI}. The statistical histograms are the simulation results obtained from 100,000 trials. The curves and histograms indicate that the Gamma approximation and the flabellate annulus approximation approach are quite accurate.

Now both the signal power $T(y)$ and the interference power $I(y)$ in \eqref{Pto} follow Gamma distributions \eqref{GamaT} and \eqref{GamaI}, which makes our analysis mathematically tractable. Consequently, the closed-form analytical results of COP is given as the following proposition by applying Corollary 1.

\textit{Proposition 2:} The COP in \eqref{Pto} is given by
\begin{align}
\mathcal{P}_{to} & =  1 - \frac{q_y^{\nu_{Ty}}\Gamma(\nu_{Ty} + \nu_{Iy})}{\nu_{Iy}(q_y + 1)^{\nu_{Ty} + \nu_{Iy}}\Gamma(\nu_{Ty})\Gamma(\nu_{Iy})}~\F\left(1,\nu_{Ty} + \nu_{Iy};\nu_{Iy} + 1;\frac{1}{q_y + 1}\right) ,\label{Pto2}
\end{align}
where $q_y = \frac{\beta\theta_{Iy}}{\theta_{Ty}}$, and ${}_2F_1\left(a,b;c;d\right)$ denotes hypergeometric function \cite[Eq. 6.455.1]{TableIntegral}.

\begin{proof}
The proof is similar to that in Section III.
\end{proof}

The theoretical results in \eqref{Pto2} and the simulation results are validated in Fig.~\ref{COP1}. The simulation results are calculated as the ratio of the number of connection outage to a total of 100,000 Monte Carlo trials. In Fig.~\ref{COP1}, although these two curves are not quite the same, the maximum discrepancy between them at $\beta = -19$ dB is smaller than 0.1. We can find that the theoretical results are very close to that of the numerical results, and the Gamma approximation is very accurate, so our theoretical results can be applied to analyze the COP.

\section{Secrecy Outage Probability}
In this section, we study the SOP in our scheme. We first  derive the closed-form expression of the SOP with a single eavesdropper. Then we discuss the SOP with multiple eavesdroppers following the PPP distribution and obtain its upper bound.

\subsection{Single Eavesdropper}
The SOP with a single eavesdropper is defined as the probability that the SIR achieved at the eavesdropper is larger than some threshold $\beta_e$ \cite{SOPbyOther}. Therefore, the SOP is given by
\begin{align}
\mathcal{P}_{so} & = \mathbb{P}\left(SIR_z>\beta_e\right) = \mathcal{P}\left(\frac{T(z)}{I(z)}>\beta_e\right) \nonumber \\
& = 1 - \mathbb{P}\left(\frac{\left|\sum_{x_R\in\Phi_R}\frac{\sqrt{P_R}H_{x_R,z}H_{x_R,y}^*}{\|H_{x_R,y}\|}d_{x_R,z}^{-\alpha/2}\right|^2}{\sum_{x_J\in\mathcal{\overline{D}}}P_jh_{x_J,z}d_{x_J,z}^{-\alpha}}\leq\beta_e\right) ,\label{Pso2}
\end{align}
where $T(z) = \left|\sum_{x_R\in\Phi_R}\frac{\sqrt{P_R}H_{x_R,z}H_{x_R,y}^*}{\|H_{x_R,y}\|}d_{x_R,z}^{-\alpha/2}\right|^2$. For arbitrary $x_R\in\Phi_R$, $\frac{H_{x_R,z}H_{x_R,y}^*}{\|H_{x_R,y}\|}d_{x_R,z}^{-\alpha/2}\thicksim \mathcal{CN}(0,d_{x_R,z}^{-\alpha})$ is an independent circularly symmetric complex Gaussian
distribution with variance $d_{x_R,z}^{-\alpha}$. Hence, $T(z)$ is conditionally exponential distributed with conditional mean $P_R\sum_{x_R\in\Phi_R}d_{x_R,z}^{-\alpha}$, which is a RV as well related to the locations of $x_R$ in $\Phi_R$.
Consequently, it is untractable to calculate \eqref{Pso2}. Fortunately, notice that $T(z)$ has an approximated Gamma distribution, and therefore we can use the DGR approach to facilitate mathematically tractable calculations of \eqref{Pso2}.

Firstly, we model the signal power $T(z)$ and the interference power $I(z)$ at the eavesdropper $z$ as Gamma variables. Due to the exponential distribution of $T(z)$, we can easily obtain the mean and the variance of $T(z)$ to derive the parameters of the Gamma model. As for the parameters of $I(z)$, we will apply \eqref{Eti} similarly to that of $I(y)$. The parameters of the PDFs of $T(z)$ (i.e. $\nu_{Tz}$, $\theta_{Iz}$) and $I(z)$ (i.e. $\nu_{Iz}$, $\theta_{Iz}$) are given as
\begin{align}
\nu_{Tz} = \frac{\lambda_RQ_z(1)}{\lambda_RQ^2_z(1)+2Q_z(2)},
\quad\theta_{Tz} = \frac{P_R\left[\lambda_RQ^2_z(1)+2Q_z(2)\right]}{Q_z(1)} 
\end{align}
and
\begin{align}
\nu_{Iz} = \frac{\lambda_J\big(\int_{\mathcal{\overline{D}}}\frac{1}{d_{x_J,z}^{\alpha}}dx_J\big)^2}{2\int_{\mathcal{\overline{D}}}\frac{1}{d_{x_J,z}^{2\alpha}}dx_J},
\quad\theta_{Iz} = \frac{2P_j\int_{\mathcal{\overline{D}}}\frac{1}{d_{x_J,z}^{2\alpha}}dx_J}{\int_{\mathcal{\overline{D}}}\frac{1}{d_{x_J,z}^{\alpha}}dx_J},
\end{align}
respectively, where $Q_z(n) = \int_{\mathcal{A}(o,L_1)}d_{x_R,z}^{-n\alpha}\mathrm{d}x_R$.
Then we obtain the approximated PDFs of $T(z)$ and $I(z)$ as
\begin{align}
f_{T(z)}(x_T;\nu_{Tz},\theta_{Tz}) = \frac{x_T^{\nu_{Tz} - 1}e^{-x_T/\theta_{Tz}}}{\theta_{Tz}^{\nu_{Tz}}\Gamma(\nu_{Tz})}
\quad\mathrm{and}\quad f_{I(z)}(x_I;\nu_{Iz},\theta_{Iz}) = \frac{x_I^{\nu_{Iz} - 1}e^{-x_I/\theta_{Iz}}}{\theta_{Iz}^{\nu_{Iz}}\Gamma(\nu_{Iz})},
\end{align}
respectively.
Therefore, according to Corollary 1, the closed-form analytical result of SOP is given by the following proposition.

\textit{Proposition 3:} The SOP in \eqref{Pso2} is given by
\begin{align}
\mathcal{P}_{so} & = \frac{q_e^{\nu_{Tz}}\Gamma(\nu_{Tz} + \nu_{Iz})}{\nu_{Iz}(q_e + 1)^{\nu_{Tz} + \nu_{Iz}}\Gamma(\nu_{Tz})\Gamma(\nu_{Iz})}~\F\left(1,\nu_{Tz} + \nu_{Iz};\nu_{Iz} + 1;\frac{1}{q_e + 1}\right) ,\label{Pso3}
\end{align}
where $q_e = \frac{\beta_e\theta_{Iz}}{\theta_{Tz}}$, and ${}_2F_1\left(a,b;c;d\right)$ denotes hypergeometric function \cite[Eq. 6.455.1]{TableIntegral}.

\begin{proof}
The proof is similar to that in Section III.
\end{proof}

\begin{figure}[!t]
\centering
\includegraphics[width=4in]{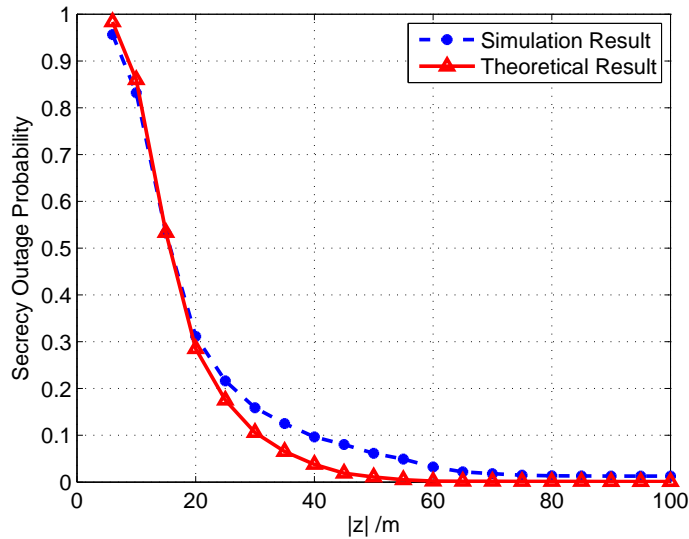}
\caption{The SOP vs. different distances $|z|$ of the eavesdropper. The system parameters are $\beta_e=0$ dB, $\alpha=4$, $L_1=6$ m, $L_2=100$ m, $d=60$ m, $L_G=5$ m, $C_1=0.8$, $C_2=0.79$, $\lambda=0.2 /\mathrm{m}^2$, $P_R=10$ dBm, and $P_j=1$ dBm.}
\label{SOP1}
\end{figure}
Fig.~\ref{SOP1} depicts the theoretical results in (\ref{Pso3}) and the simulation results of the SOP, where $|z|$ denotes the distance between the source and the eavesdropper. 100,000 Monte Carlo trials are used. From Fig.~\ref{SOP1}, we can observe that although the two curves are not quite the same, the maximum discrepancy between them at $|z| = 45$ m is smaller than 0.1. We can find that the theoretical curves coincide with the simulation ones well, and the Gamma approximation is close to our system model, which validates our theoretical results in Proposition 3.

\subsection{Multiple Eavesdroppers}
When there are multiple eavesdroppers located at the annulus $\mathcal{D}(L_1,L_2)$ in the network, we assume that they are modeled as a homogeneous PPP $\Phi_E$ with density $\lambda_e$. The SOP with multiple eavesdroppers is defined as the probability that the SIR achieved by anyone of the eavesdroppers is larger than some threshold $\beta_e$. Therefore, the SOP is given by
\begin{align}
\mathcal{P}_{so} & = \mathbb{P}\left(\cup_{z\in\Phi_E}SIR_z>\beta_e\right) \nonumber \\
& = 1 - \mathbb{P}\left(\cap_{z\in\Phi_E}SIR_z\leq\beta_e\right) \nonumber \\
& = 1 - \mathbb{E}_{\Phi_R,\Phi_J,\Phi_E}\left[\mathbb{P}\left(\cap_{z\in\Phi_E}\frac{T(z)}{I(z)}\leq\beta_e\Bigg|\Phi_R,\Phi_J,\Phi_E\right)\right] .\label{Pso}
\end{align}
From \eqref{Pso} we can see that, it is hard to obtain an exact closed-form expression of the SOP with three independent and homogeneous PPPs ($\Phi_R,\Phi_J,\Phi_E$) mathematically. To achieve tractable and accurate results, we make a compromise to obtain the upper bound of the SOP.

We first focus on the calculation over $\Phi_E$ as
\begin{align}
\mathcal{P}_{so} & = 1 - \mathbb{E}_{\Phi_R,\Phi_J,\Phi_E}\left[\prod_{z\in\Phi_E}\mathbb{P}\left(SIR_z\leq\beta_e\Big|\Phi_R,\Phi_J,\Phi_E\right)\right] \nonumber \\
& \overset{(a)}{=} 1 - \mathbb{E}_{\Phi_R,\Phi_J}\left[\exp\left(-\lambda_e\int_{\mathcal{D}(L_1,L_2)}\mathbb{P}\left(SIR_z>\beta_e\Big|\Phi_R,\Phi_J\right)dz\right)\right] \nonumber \\
& \overset{(b)}{\leq} 1 - \exp\left[-\lambda_e\int_{\mathcal{D}(L_1,L_2)}\mathbb{E}_{\Phi_R,\Phi_J}\left[\mathbb{P}\left(\frac{T(z)}{I(z)}>\beta_e\Big|\Phi_R,\Phi_J\right)dz\right]\right] ,\label{PsoJ}
\end{align}
where $(a)$ follows from the probability generating functional (PGFL) of the PPP, and $(b)$ is derived by applying the Jensen's Inequality. We define \eqref{PsoJ} as $\overline{\mathcal{P}}_{so}$, i.e., the upper bound of $\mathcal{P}_{so}$. Since $T(z) = \big|\sum_{x_R\in\Phi_R}\frac{\sqrt{P_R}H_{x_R,z}H_{x_R,y}^*}{\|H_{x_R,y}\|}d_{x_R,z}^{-\alpha/2}\big|^2$ is an exponential random variable with parameter $\sum_{x_R\in\Phi_R}P_Rd_{x_R,z}^{-\alpha}$, we can obtain
\begin{align}
\overline{\mathcal{P}}_{so} & = 1 - \exp\left[-\lambda_e\int_{\mathcal{D}(L_1,L_2)}\mathbb{E}_{\Phi_R,\Phi_J}\left[e^{-\frac{\beta_eI(z)}{\sum_{x_R\in\Phi_R}P_Rd_{x_R,z}^{-\alpha}}}\right]dz\right].
\end{align}

\begin{figure}[!t]
\centering
\includegraphics[width=4in]{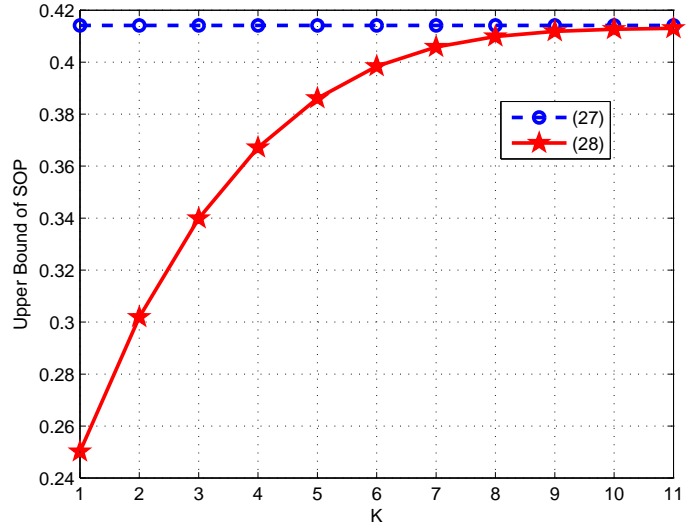}
\caption{Convergence results of \eqref{PsoUB2} for different $K$ when $\beta_e = 0$ dB. The system parameters are $\alpha=4$, $L_1=6$ m, $L_2=100$ m, $d=60$ m, $L_G=5$ m, $C_1=0.8$, $C_2=0.79$, $\lambda=0.2 /\mathrm{m}^2$, $\lambda_E=0.0005 /\mathrm{m}^2$, $P_R=10$ dBm, and $P_j=1$ dBm.}
\label{Kvalue}
\end{figure}

\begin{figure}[!t]
\centering
\includegraphics[width=4in]{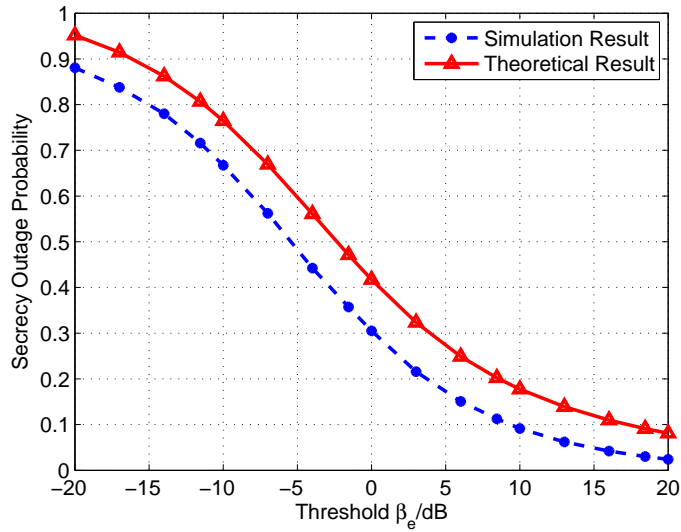}
\caption{The SOP vs. $\beta_e$ with PPP distributed eavesdroppers. The system parameters are $\alpha=4$, $L_1=6$ m, $L_2=100$ m, $d=60$ m, $L_G=5$ m, $C_1=0.8$, $C_2=0.79$, $\lambda=0.2 /\mathrm{m}^2$, $\lambda_E=0.0005 /\mathrm{m}^2$, $P_R=10$ dBm, and $P_j=1$ dBm.}
\label{SOP2}
\end{figure}

Substituting $I(z) = \sum_{x_J\in\Phi_J}P_jh_{x_J,z}d_{x_J,z}^{-\alpha}$ yields
\begin{align}
\overline{\mathcal{P}}_{so} & = 1 - \exp\left[-\lambda_e\int_{\mathcal{D}(L_1,L_2)}\mathbb{E}_{\Phi_R,\Phi_J}\left[e^{-\frac{\beta_e\sum_{x_J\in\Phi_J}P_jh_{x_J,z}d_{x_J,z}^{-\alpha}}{\sum_{x_R\in\Phi_R}P_Rd_{x_R,z}^{-\alpha}}}\right]dz\right] \nonumber \\
& = 1 - \exp\left[-\lambda_e\int_{\mathcal{D}(L_1,L_2)}\mathbb{E}_{\Phi_R,\Phi_J}\left[\prod_{x_J\in\Phi_J}e^{-\frac{\beta_eP_jh_{x_J,z}d_{x_J,z}^{-\alpha}}{\sum_{x_R\in\Phi_R}P_Rd_{x_R,z}^{-\alpha}}}\right]dz\right] \nonumber
\end{align}
\begin{align}
& \overset{(a)}{=} 1 - \exp\left[-\lambda_e\int_{\mathcal{D}(L_1,L_2)}\mathbb{E}_{\Phi_R}\left[\exp\left(-\lambda_J\int_{\mathcal{\overline{D}}}\mathbb{E}_{h_{x_J,z}}\left(1 - e^{-\frac{\beta_eP_jh_{x_J,z}d_{x_J,z}^{-\alpha}}{\sum_{x_R\in\Phi_R}P_Rd_{x_R,z}^{-\alpha}}}\right)dx_J\right)\right]dz\right] \nonumber \\
& \overset{(b)}{=} 1 - \exp\left[-\lambda_e\int_{\mathcal{D}(L_1,L_2)}\mathbb{E}_{\Phi_R}\left[\exp\left(-\lambda_J\int_{\mathcal{\overline{D}}}\frac{\beta_eP_jd_{x_J,z}^{-\alpha}}{\sum_{x_R\in\Phi_R}P_Rd_{x_R,z}^{-\alpha} + \beta_eP_jd_{x_J,z}^{-\alpha}}dx_J\right)\right]dz\right] ,\label{PsoUB1}
\end{align}
where $(a)$ follows from applying the PGFL of the PPP, since that the locations $x_J$ of the jammers are PPP distributed, and $(b)$ follows from $h_{x_J,z}\thicksim\exp(1)$.

In \eqref{PsoUB1}, there is still a PPP $\Phi_R$ for the relays, which means a triple integral is required to be performed. This leads to an unacceptable calculation burden. However, since the number of the relays is Poisson distributed, we can take the discrete expectation to approximate the continuous expectation. In this case, the law of total probability is employed for $\mathbb{E}_{\Phi_R}[\cdot]$, then the $\overline{\mathcal{P}}_{so}$ is given by
\begin{align}
1 - \exp\left[-\lambda_e\int_{\mathcal{D}(L_1,L_2)}\sum_{k=1}^K\frac{e^{-\lambda_R}\lambda_R^k}{k!}\exp\left(-\lambda_J\int_{\mathcal{\overline{D}}}\frac{\beta_eP_jd_{x_J,z}^{-\alpha}}{\sum_{x_R\in\Phi_R}P_Rd_{x_R,z}^{-\alpha} + \beta_eP_jd_{x_J,z}^{-\alpha}}dx_J\right)dz\right] ,\label{PsoUB2}
\end{align}
where $K$ is a specific number of the relays and $\frac{e^{-\lambda_R}\lambda_R^k}{k!}$ is the probability of $k$ relays. Consequently, the analysis is simplified through such an approximation.

Fig.~\ref{Kvalue} illustrates the typical convergence behavior of \eqref{PsoUB2} as a function of $K$. The dash curve representing the Monte Carlo simulation results of \eqref{PsoUB1} with 100,000 trials. The solid curve is the numerical results of \eqref{PsoUB2}. From the figure we can see that \eqref{PsoUB2} is convergent to \eqref{PsoUB1} which is the upper bound of the SOP and stabilizes after $K = 10$. The convergence rate decreases with increasing $K$. The $K$ which stabilizes \eqref{PsoUB2} is related to the value of $\frac{e^{-\lambda_R}\lambda_R^k}{k!}$.

The theoretical results in (\ref{PsoUB2}) and the simulation results are validated in Fig.~\ref{SOP2}. $K$ is set as 11 due to the analysis of convergence in Fig.~\ref{Kvalue}. From \eqref{PsoJ}, \eqref{PsoUB1} is the upper bound of the SOP due to utilizing of the Jensen's Inequality. It is obvious that \eqref{PsoUB2} convergent to \eqref{PsoUB1} is the upper bound. We observe from the figure that, the theoretical $\overline{\mathcal{P}}_{so}$ is verified by the Monte Carlo simulation and SOP in (\ref{PsoUB2}) is the upper bound. The gap between these two curves is due to the following two reasons: 1) the utilizing of the Jensen¡¯s Inequality in \eqref{PsoJ}; 2) the utilizing of the discrete expectation to approximate the continuous expectation in \eqref{PsoUB2}.

\begin{figure}[!tp]
\centering
\includegraphics[width=4in]{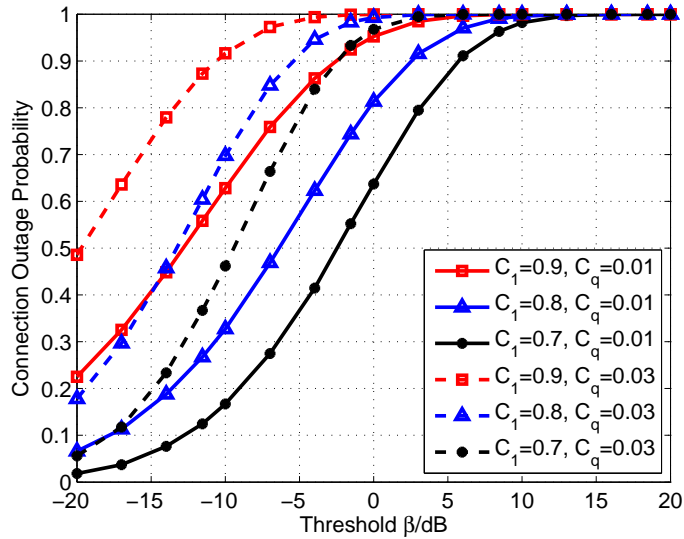}
\caption{The COP vs. $\beta$ for various social trust degrees $C_1$ and $C_q$ of the source.}
\label{top_lam_lmJ}
\end{figure}

\section{Numerical Results and Discussions}
In this section, numerical results are presented to illustrate the performance of the proposed scheme. Considering the accuracy of the Gamma approximation and the DGR approach, which have been validated in Figs.~\ref{GammaModel}--\ref{SOP1}, we only present the theoretical results based on Propositions 2 and 3 for simplicity. As the legitimate nodes in the network are categorized into relays or jammers according to their social trust degrees of the source, we firstly focus on the impacts of $C_1$ and $C_q$ on the performance of the networks, where $C_q \triangleq (C_1 - C_2)$. The benchmark scheme is that the relays transmit without the assistance of jammers, which is adopted as no jammer assistance (NJA) scheme. Then the secrecy performance is illustrated versus various values of $L_1$ and $L_G$. Finally, the upper bound versus various parameters is provided for the SOP of multiple eavesdroppers. The system parameters are set as the followings unless otherwise noted: $\alpha=4$, $L_1=6$ m, $L_2=100$ m, $d=60$ m, $L_G=5$ m, $\lambda=0.2 /\mathrm{m}^2$, $C_1 = 0.8$, $C_q = 0.01$, $P_R=10$ dBm, and $P_j=1$ dBm.

Fig.~\ref{top_lam_lmJ} illustrates the COP vs. $\beta$ for various social trust degrees $C_1$ and $C_q$ of the source. We observe that for the same $C_1$, a smaller $C_q$ leads to a lower COP. The reason is that $C_1$ and $C_q$ determine $\lambda_R$ and $\lambda_J$. A smaller $\lambda_J$ is equivalent to a lower interference power, which results in a lower COP. We also see that for the same $C_q$, the COP increases with the increasing $C_1$. This is because a smaller $\lambda_R$ produces fewer relays leading to a lower SIR.


\begin{figure}[!tp]
\centering
\subfigure[]{\label{SOP_C1}
\includegraphics[width=3.5in]{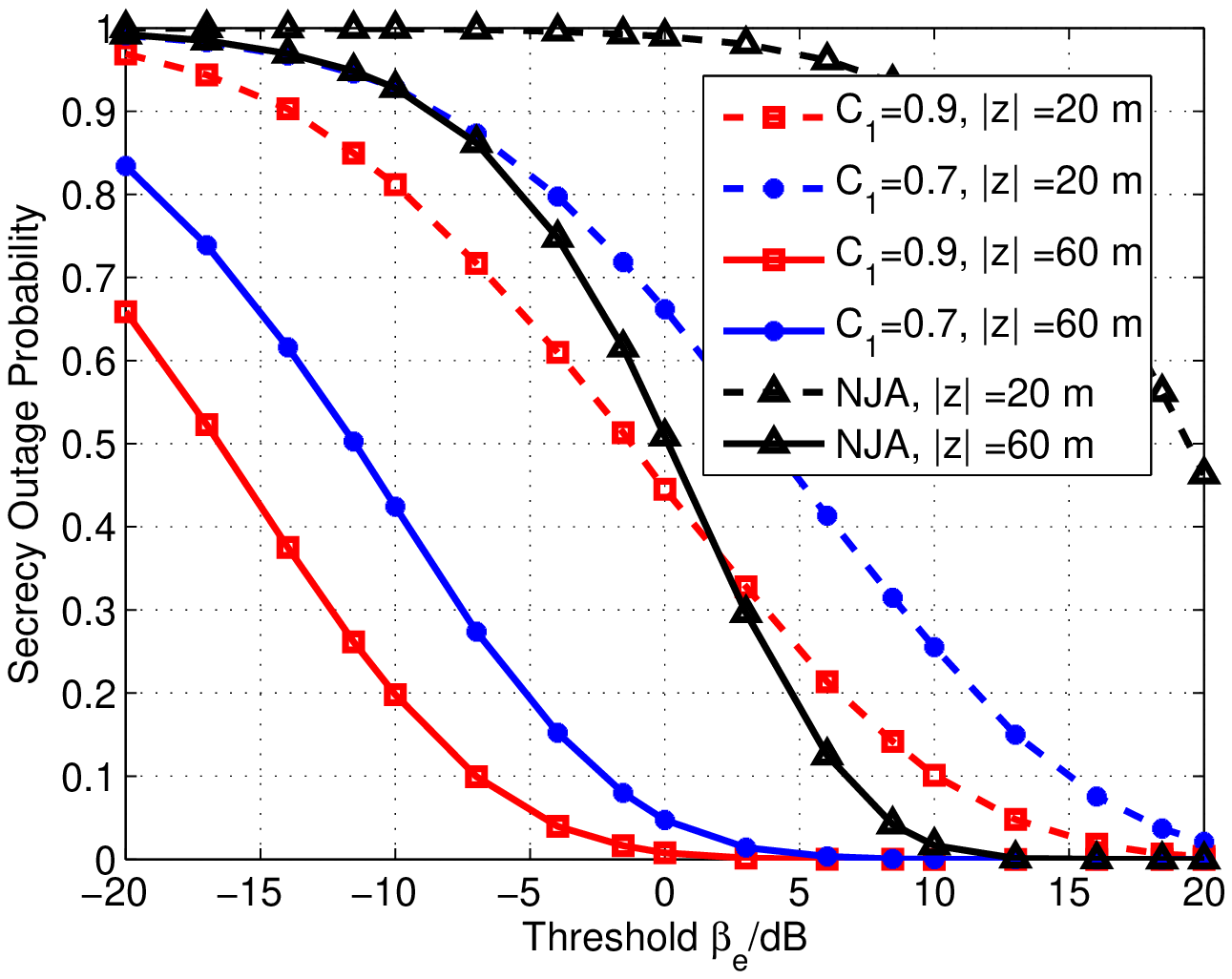}}
\subfigure[]{\label{SOP_Cq}
\includegraphics[width=3.5in]{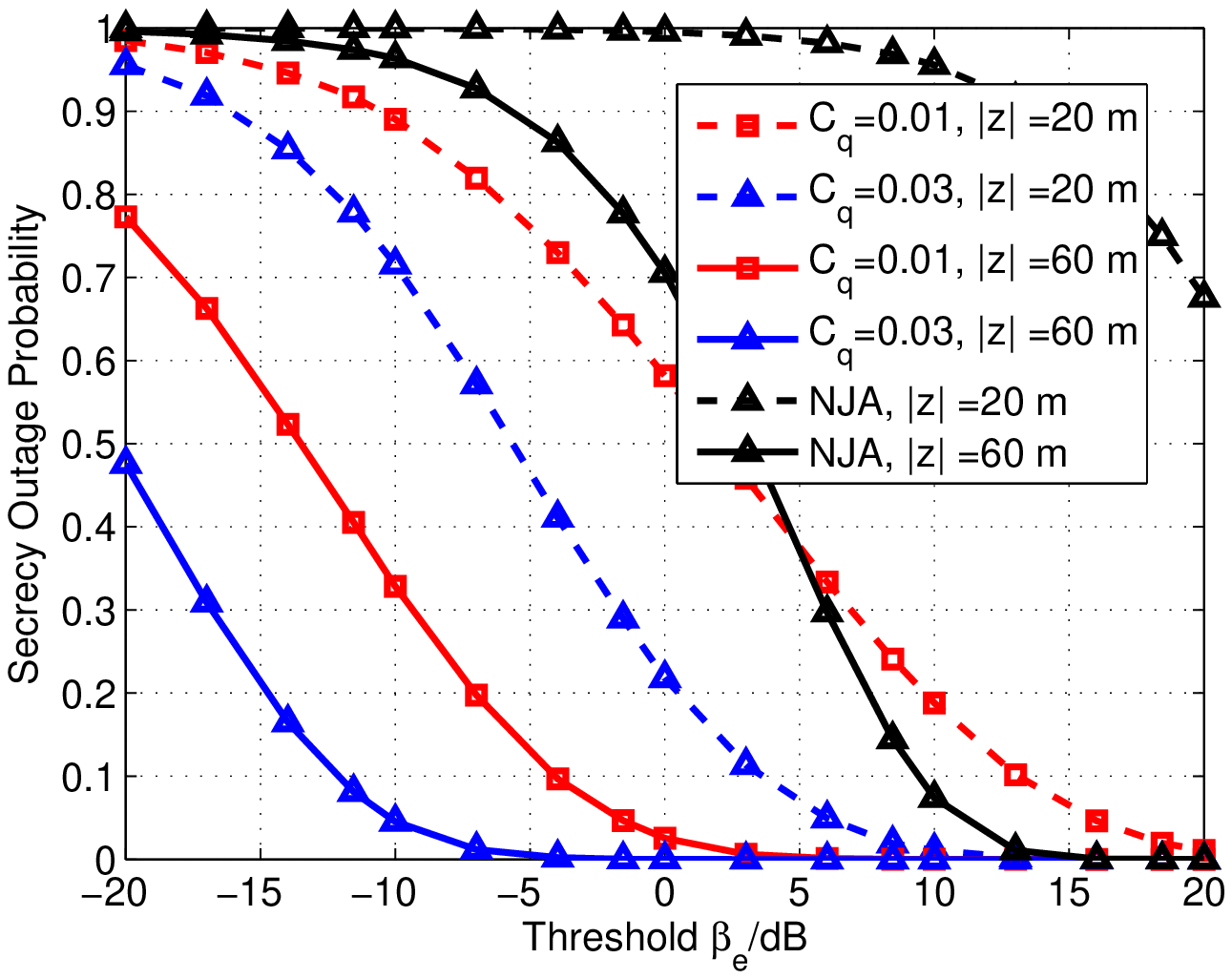}}
\caption{ The SOP of a single eavesdropper vs. $\beta_e$ for various $C_1$ and $C_q$ of the source and distances $|z|$ of the eavesdropper.: (a) SOP with different $C_1$ and $|z|$; (b) SOP with different $C_q$ and $|z|$.}
\label{SOP_CC}
\end{figure}


Fig.~\ref{SOP_C1} plots the SOP versus $\beta_e$ for various social trust degrees $C_1$ of the source and distances $|z|$ of the eavesdropper. Comparing the curves with the same $|z|$, we see that as $C_1$ increases, the SOP decreases. This is because a larger $C_1$ is equivalent to a smaller $\lambda_R$, which results in less relays and produce a lower $SIR_E$. Therefore, there is a higher probability for performing perfect secrecy, which leads to a lower SOP. We also see that the SOP decreases with increasing $|z|$ by the comparison among the curves with the same $C_1$. This is due to the fact that the secrecy outage occurs more frequently when the distance between the source and the eavesdropper decreases. From this figure, the proposed scheme has better performance than that of the NJA scheme at $|z| = 20$ m and $|z| = 60$ m, respectively. Since $C_1$ represents the trust degree of the source, we know that the most private message should be transmitted to the person with a sufficiently high trust degree in order to realize perfect secrecy.


Fig.~\ref{SOP_Cq} compares the SOP versus $\beta_e$ for various $C_q$ and $|z|$. By observing the curves with the same $|z|$, we see that as $C_q$ increases, the SOP decreases. This is because a larger $C_q$ is equivalent to a larger $\lambda_J$, which results in more jammers producing lower $SIR_E$. We also observe that the SOP dramatically decreases with increasing $|z|$. As $C_q$ determines the density of the jammers, smaller $C_2$ will lead to more jammers offering intentional interference to improve the security performance. Also, the proposed scheme performs much better than that of the NJA scheme in both cases with $|z| = 20$ m and $|z| = 60$ m. From this we know that a diminishing social trust degree will disrupt the eavesdropper more efficiently.

%


\begin{figure}[!tp]
\centering
\includegraphics[width=4in]{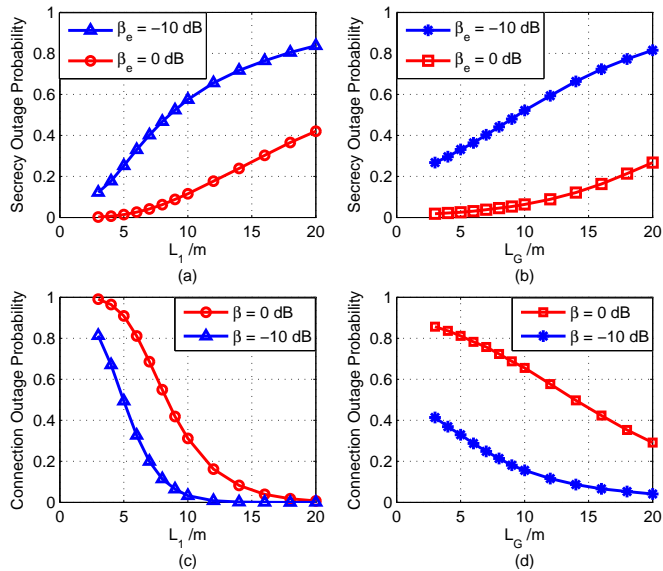}
\caption{The SOP of a single eavesdropper and the COP vs. $L_1$ and $L_G$ for different $\beta_e$ and $\beta$, respectively: (a) SOP vs. $L_1$; (b) SOP vs. $L_G$; (c) COP vs. $L_1$; (d) COP vs. $L_G$.}
\label{SOP_LL}
\end{figure}

Fig.~\ref{SOP_LL} illustrates the SOP of a single eavesdropper and the COP versus $L_1$ and $L_G$ for different $\beta_e$ and $\beta$, respectively. Fig.~\ref{SOP_LL}(a) depicts that the SOP increases with increasing $L_1$. This is because the secrecy outage occurs more frequently when $L_1$ increases producing more relays. From Fig.~\ref{SOP_LL}(b), we can observe that when $L_G$ increases, the SOP increases. This is due to the fact that a larger $L_G$ leads to less interference causing by the jammers near to the eavesdropper, which leads to a higher $SIR_E$. We also see that the SOP has dramatic increasing at a smaller threshold $\beta_e$. Similarly, the COP decreases with increasing $L_1$ and $L_G$ in Fig.~\ref{SOP_LL}(c) and Fig.~\ref{SOP_LL}(d), respectively. Such comparison on the SOP and the COP helps the trade-off between $L_1$ and $L_G$.

\begin{figure}[!tp]
\centering
\includegraphics[width=4in]{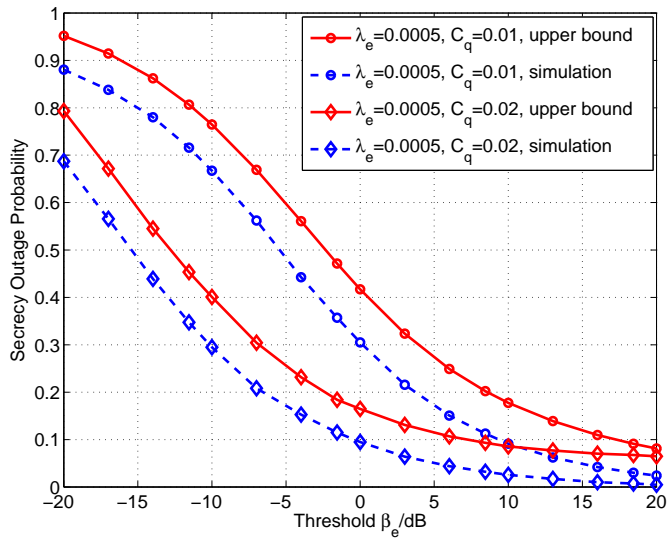}
\caption{The SOP with multiple eavesdroppers vs. $\beta_e$ for different $\lambda_e$ and $C_q$.}
\label{mSOP_add}
\end{figure}

Finally, Fig.~\ref{mSOP_add} depicts the SOP with multiple eavesdroppers versus $\beta_e$ for different $\lambda_e$ and $C_q$. The dash curves are the Monte Carlo simulation results, while the solid curves are the numerical results of \eqref{PsoUB2}. For various $\lambda_e$ and $C_q$, the numerical results are the upper bounds which validates the analysis results in \eqref{PsoUB2}. By comparing the curves with the same $C_q$, we see that the SOP increases with increasing density of the eavesdroppers. In the comparison of the curves with the same $\lambda_e$, the SOP decreases when increasing the density of the jammers, i.e., better secrecy performance can be achieved by increasing jammers.

\section{Conclusions}
In this paper, we proposed a cooperative relay and jamming scheme based on the social trust degrees to secure communications. The security performance is investigated in terms of the COP and the SOP under a stochastic geometry framework. A DGR approach was proposed to facilitate the analysis of these metrics and closed-form expressions were obtained. The simulation results highlighted that the social trust degrees have dramatic influences on the security performance in the networks. For example, the private message should be transmitted to the person with sufficiently high trust degree in order to realize secure communications, meanwhile a diminishing social trust degree will disrupt the eavesdropper more efficiently. In addition, the protected zone can protect communications efficiently when the eavesdropper is away from the source. As an extension, we further investigated the SOP in the presence of PPP distributed eavesdroppers and obtained its upper bound. Such a scenario has practical interest since can be implemented in offices, laboratories, and dormitories, where the social trust degree is employed to reflect the willingness of cooperation of the users.

\appendices
\section{Proof of $\nu_{Ty}$ and $\theta_{Ty}$ in Proposition 1}
According to \eqref{Eti} and the definition of $T(y)$ in \eqref{Ty}, the $i$-th cumulants of $T(y)$ is given by
\begin{align}
N_T^{(i)} & = \frac{d^{2i}\mathbb{E}_{\Phi_R,\|h_{x_R,y}\|}\left[e^{w\sum_{x_R\in\Phi_R}\sqrt{P_R}\|h_{x_R,y}\|d_{x_R,y}^{-\alpha/2}}\right]}{dw^{2i}}\Big|_{w = 0} .\label{ETy}
\end{align}
First, we calculate $\mathbb{E}_{\Phi_R,\|h_{x_R,y}\|}\left[e^{w\sum_{x_R\in\Phi_R}\sqrt{P_R}\|h_{x_R,y}\|d_{x_R,y}^{-\alpha/2}}\right]$ as
\begin{align}
&\quad\mathbb{E}_{\Phi_R,\|h_{x_R,y}\|}\left[e^{w\sum_{x_R\in\Phi_R}\sqrt{P_R}\|h_{x_R,y}\|d_{x_R,y}^{-\alpha/2}}\right] \nonumber \\
& = \mathbb{E}_{\Phi_R,\|h_{x_R,y}\|}\bigg[\prod_{x_R\in\Phi_R}e^{w\sqrt{P_R}\|h_{x_R,y}\|d_{x_R,y}^{-\alpha/2}}\bigg] \nonumber \\
& \overset{(a)}{=} \mathbb{E}_{\Phi_R}\bigg[\prod_{x_R\in\Phi_R}\mathbb{E}_{\|h_{x_R,y}\|}\left[e^{w\sqrt{P_R}\|h_{x_R,y}\|d_{x_R,y}^{-\alpha/2}}\right]\bigg] \nonumber \\
& \overset{(b)}{=} \mathbb{E}_{\Phi_R}\Big[\prod_{x_R\in\Phi_R}Ge^{Qw^2}\Big] \nonumber \\
& \overset{(c)}{=} \exp\left[\lambda_R\int_{\mathcal{A}_L}\Big(Ge^{Qw^2}-1\Big)dx_R\right], \label{AppA1} 
\end{align}
where $Q = \frac{1}{2}P_Rd_{x_R,y}^{-\alpha}$, $G = \Gamma(1,Qw^2)+\sqrt{Q}w\Gamma(\frac{1}{2},Qw^2)$, $\Gamma(a,b) = \int_b^{\infty}t^{a-1}e^{-t}dt$ and $\mathcal{A}_L$ denotes the area $\mathcal{A}(o,L_1)$. $(a)$ follows since $\|h_{x_R,y}\|$ is independent of $\Phi_R$. $(b)$ follows since $\|h_{x_R,y}\|$ is Rayleigh distributed, and by applying the PGFL of the PPP we can obtain $(c)$. Now \eqref{ETy} is equal to
\begin{align}
N_T^{(i)} & = 
\frac{d^{2i}\exp\left[\lambda_R\int_{\mathcal{A}_L}\Big(Ge^{Qw^2}-1\Big)dx_R\right]}{dw^{2i}}\bigg|_{w = 0} .\label{ETy2}
\end{align}
Consequently, $N_T^{(1)}$ is given by
\begin{align}
N_T^{(1)} & = \frac{d^{2}\exp\left[\lambda_R\int_{\mathcal{A}_L}\Big(Ge^{Qw^2}-1\Big)dx_R\right]}{dw^2}\Big|_{w = 0} \nonumber \\
& \overset{(a)}{=} \bigg\{\exp\left[\lambda_R\int_{\mathcal{A}_L}\Big(Ge^{Qw^2}-1\Big)dx_R\right]\left[\lambda_R\int_{\mathcal{A}_L}\left(2QwGe^{Qw^2}+4Qw\right)dx_R\right]^2 \nonumber \\
&\quad + \exp\left[\lambda_R\int_{\mathcal{A}_L}\Big(Ge^{Qw^2}-1\Big)dx_R\right]\left[\lambda_R\int_{\mathcal{A}_L}2Q\left(1+2Qw^2\right)\left(2+Ge^{Qw^2}\right)dx_R\right]\bigg\}\bigg|_{w = 0} \nonumber \\
& \overset{(b)}{=} 3\lambda_R\int_{\mathcal{A}_L}\frac{P_R}{d_{x_R,y}^{\alpha}}dx_R, \label{AppA2}
\end{align}
where $(a)$ is the second-order differential results, and let $w = 0$, we can obtain $(b)$.
Let $i = 2$, and $N_T^{(2)}$ is given by
\begin{align}
N_T^{(2)} & = \frac{d^{4}\exp\left[\lambda_R\int_{\mathcal{A}_L}\Big(Ge^{Qw^2}-1\Big)dx_R\right]}{dw^4}\Big|_{w = 0} \nonumber \\
& \overset{(a)}{=} 54\lambda_R^2\left(\int_{\mathcal{A}_L}\frac{P_R}{d_{x_R,y}^{\alpha}}dx_R\right)^2 + 9\lambda_R\int_{\mathcal{A}_L}\frac{P_R^2}{d_{x_R,y}^{2\alpha}}dx_R, \label{AppA3} 
\end{align}
where $(a)$ is the fourth-order differential results and $w = 0$. Consequently, by substituting \eqref{AppA2} and \eqref{AppA3}, $\sigma_T^2$ is given as
\begin{align}
\sigma_T^2 & = N_T^{(2)} - \left(N_T^{(1)}\right)^2 \nonumber \\
& = 45\lambda_R^2\left(\int_{\mathcal{A}_L}\frac{P_R}{d_{x_R,y}^{\alpha}}dx_R\right)^2 + 9\lambda_R\int_{\mathcal{A}_L}\frac{P_R^2}{d_{x_R,y}^{2\alpha}}dx_R. \label{AppA4}
\end{align}
As a result, we obtain $\nu_{Ty}$ and $\theta_{Ty}$ by substituting \eqref{AppA2} and \eqref{AppA4} due to \eqref{vt} as
\begin{align}
\nu_{Ty} & = \frac{\lambda_R\left(\int_{\mathcal{A}_L}d_{x_R,y}^{-\alpha}dx_R\right)^2}{5\lambda_R\left(\int_{\mathcal{A}_L}d_{x_R,y}^{-\alpha}dx_R\right)^2+\int_{\mathcal{A}_L}d_{x_R,y}^{-2\alpha}dx_R} \nonumber \\
& = \frac{\lambda_RQ^2_y(1)}{5\lambda_RQ^2_y(1)+Q_y(2)}
\end{align}
and
\begin{align}
\theta_{Ty} & = \frac{15\lambda_RP_R\left(\int_{\mathcal{A}_L}d_{x_R,y}^{-\alpha}dx_R\right)^2+3P_R\int_{\mathcal{A}_L}d_{x_R,y}^{-2\alpha}dx_R}{\int_{\mathcal{A}_L}d_{x_R,y}^{-\alpha}dx_R} \nonumber \\
& = \frac{3P_R\left[5\lambda_RQ^2_y(1)+Q_y(2)\right]}{Q_y(1)},
\end{align}
respectively.

\section{Proof of $\nu_{Iy}$ and $\theta_{Iy}$ in Proposition 1}
According to \eqref{Eti} and the definition of $I(y)$ in \eqref{Iy}, the $i$-th cumulants of $I(y)$ is given by
\begin{align}
N_I^{(i)} & = \frac{d^i\mathbb{E}_{\Phi_J,h_{x_J,y}}\left[e^{w\sum_{x_J\in\mathcal{\overline{D}}}P_jh_{x_J,y}d_{x_J,y}^{-\alpha}}\right]}{dw^i}\Big|_{w = 0} .\label{EIy}
\end{align}
Similar to the derivation of \eqref{AppA1}, $\mathbb{E}_{\Phi_J,h_{x_J,y}}\left[e^{w\sum_{x_J\in\mathcal{\overline{D}}}P_jh_{x_J,y}d_{x_J,y}^{-\alpha}}\right]$ is given as
\begin{align}
&\quad \mathbb{E}_{\Phi_J,h_{x_J,y}}\left[e^{w\sum_{x_J\in\mathcal{\overline{D}}}P_jh_{x_J,y}d_{x_J,y}^{-\alpha}}\right] = \exp\left(\lambda_J\int_{\mathcal{\overline{D}}}\frac{wP_j}{d_{x_J,y}^{\alpha} - wP_j}dx_J\right). \label{AppB1}
\end{align}
Consequently, substituting \eqref{AppB1} into \eqref{EIy}, $N_I^{(1)}$ is given by
\begin{align}
N_I^{(1)} & = \frac{d\Big(\exp\big(\lambda_J\int_{\mathcal{\overline{D}}}\frac{wP_j}{d_{x_J,y}^{\alpha} - wP_j}dx_J\big)\Big)}{dw}\Big|_{w = 0} \nonumber \\
& \overset{(a)}{=} \exp\bigg(\lambda_J\int_{\mathcal{\overline{D}}}\frac{wP_j}{d_{x_J,y}^{\alpha} - wP_j}dx_J\bigg)\lambda_J\int_{\mathcal{\overline{D}}}\frac{P_jd_{x_J,y}^{\alpha}}{\left(d_{x_J,y}^{\alpha} - wP_j\right)^2}dx_J\bigg|_{w = 0} \nonumber \\
& \overset{(b)}{=} \lambda_J\int_{\mathcal{\overline{D}}}\frac{P_j}{d_{x_J,y}^{\alpha}}dx_J, \label{AppB2}
\end{align}
where $(a)$ is the derivation results, and let $w=0$, we have $(b)$. Let $i = 2$, $N_I^{(2)}$ is given as
\begin{align}
N_I^{(2)} & = \frac{d^2\Big(\exp\big(\lambda_J\int_{\mathcal{\overline{D}}}\frac{wP_j}{d_{x_J,y}^{\alpha} - wP_j}dx_J\big)\Big)}{dw^2}\Big|_{w = 0} \nonumber \\
& \overset{(a)}{=} \bigg[\exp\left(\lambda_J\int_{\mathcal{\overline{D}}}\frac{wP_j}{d_{x_J,y}^{\alpha} - wP_j}dx_J\right)\lambda_J^2\left(\int_{\mathcal{\overline{D}}}\frac{P_jd_{x_J,y}^{\alpha}}{\left(d_{x_J,y}^{\alpha} - wP_j\right)^2}dx_J\right)^2 \nonumber \\
&\quad + \exp\left(\lambda_J\int_{\mathcal{\overline{D}}}\frac{wP_j}{d_{x_J,y}^{\alpha} - wP_j}dx_J\right)\lambda_J\int_{\mathcal{\overline{D}}}\frac{2\left(d_{x_J,y}^{\alpha} - wP_j\right)P_j^2d_{x_J,y}^{\alpha}}{\left(d_{x_J,y}^{\alpha} - wP_j\right)^4}dx_J\bigg]\bigg|_{w = 0} \nonumber \\
& \overset{(b)}{=} \lambda_J^2\bigg(\int_{\mathcal{\overline{D}}}\frac{P_j}{d_{x_J,y}^{\alpha}}dx_J\bigg)^2 + 2\lambda_J\int_{\mathcal{\overline{D}}}\frac{P_j^2}{d_{x_J,y}^{2\alpha}}dx_J, \label{AppB3}
\end{align}
where $(a)$ is the second-order differential results. Let $w = 0$, and we can obtain $(b)$. As a result, by substituting \eqref{AppB2} and \eqref{AppB3}, $\sigma_I^2$ is given by
\begin{align}
\sigma_I^2 & = N_I^{(2)} - \left(N_I^{(1)}\right)^2 \nonumber \\
& = \lambda_J^2\bigg(\int_{\mathcal{\overline{D}}}\frac{P_j}{d_{x_J,y}^{\alpha}}dx_J\bigg)^2 + 2\lambda_J\int_{\mathcal{\overline{D}}}\frac{P_j^2}{d_{x_J,y}^{2\alpha}}dx_J - \bigg(\lambda_J\int_{\mathcal{\overline{D}}}\frac{P_j}{d_{x_J,y}^{\alpha}}dx_J\bigg)^2 \nonumber \\
& = 2\lambda_J\int_{\mathcal{\overline{D}}}\frac{P_j^2}{d_{x_J,y}^{2\alpha}}dx_J. \label{AppB4}
\end{align}

As a result, we obtain $\nu_{Iy}$ and $\theta_{Iy}$ by substituting \eqref{AppB2} and \eqref{AppB4} due to \eqref{vt} as
\begin{align}
\nu_{Iy} = \frac{\lambda_J\Big(\int_\mathcal{\overline{D}}\frac{1}{d_{x_J,y}^{\alpha}}dx_J\Big)^2}{2\int_\mathcal{\overline{D}}\frac{1}{d_{x_J,y}^{2\alpha}}dx_J}\quad\mathrm{and}\quad
\theta_{Iy} = \frac{2P_j\int_\mathcal{\overline{D}}\frac{1}{d_{x_J,y}^{2\alpha}}dx_J}{\int_\mathcal{\overline{D}}\frac{1}{d_{x_J,y}^{\alpha}}dx_J},
\end{align}
respectively.


\linespread{1.7}

\end{document}